\documentclass[a4paper, 10pt]{article}
\usepackage{amssymb,amsfonts,amsmath,amsthm}
\usepackage{epsfig}
\usepackage{pst-poly}     % From pstricks/contrib/pst-poly
\usepackage{pst-all}      % From PSTricks
\usepackage{amssymb}

\usepackage{mathrsfs}%有更多的数学符号，比如花体 $\mathscr P$

\usepackage{latexsym}
\usepackage[numbers,sort&compress]{natbib}
\usepackage{graphicx}
\usepackage[centerlast]{caption2}

\newtheorem{definition}{Definition}

\newtheorem{lemma}{Lemma}
\newtheorem{theorem}{Theorem}

\newtheorem{observation}{Observation}

\numberwithin{figure}{section} \numberwithin{definition}{section}
\numberwithin{observation}{section} \numberwithin{lemma}{section}
\numberwithin{theorem}{section} \numberwithin{proposition}{section}
\numberwithin{conjecture}{section} \numberwithin{table}{section}

\makeatletter \@addtoreset{equation}{section} \makeatother

\begin{document}

\title{
{An upper bound for the crossing number of bubble-sort graph
$B_n$}\footnote{The research is supported by NSFC of China
(No.60973014,\ 61170303)}}
\author{
Baigong Zheng, \ Yuansheng Yang\footnote {corresponding
author's email : yangys@dlut.edu.cn}, \ Xirong Xu \\
School of Computer Science and Technology\\
Dalian University of Technology\\ Dalian, 116024, P.R. China }

\date{}
\maketitle
\begin{abstract}
The {\it crossing number} of a graph $G$ is the minimum number of
pairwise intersections of edges in a drawing of $G$. Motivated by
the recent work [Faria, L., Figueiredo, C.M.H. de, S\'{y}kora, O.,
Vrt'o, I.: An improved upper bound on the crossing number of the
hypercube. J. Graph Theory 59, 145--161 (2008)], we give an upper
bound of the crossing number of $n$-dimensional bubble-sort graph
$B_n$.
\bigskip

\noindent {\bf Keywords:} {\it Crossing number}; {\it Drawing}; {\it
Bubble-sort graph}
\end{abstract}

\section{Introduction}
\indent \indent Symmetric graphs, such as the $n$-dimensional
Boolean hypercube $Q_n$, and the cube-connected cycles, have been
widely used as processor or communication interconnection networks.
For designing, analyzing and improving such networks, Akers and
Krishnamurthy proposed a formal group-theoretic model in
\cite{AK89}. Based on the model, the $n$-dimensional bubble-sort
graph is proposed as one of those classes of networks that have
better performance, as measured by diameter, connectivity, fault
tolerance, etc, than the popular $Q_n$. Thus, it has drawn a great
deal of attention of research (see \cite{AK07,SLHTH10,CLS10}).

The {\it n-dimensional bubble-sort graph} $B_n$ has $n!$ vertices
labeled by distinct permutations on $\{1,2,\cdots,n\}$. Let
$x=x_1x_2\cdots x_n$ be the vertex of $B_n$. Two vertices
$x=x_1x_2\cdots x_n$ and $y=y_1y_2\cdots y_n$ in $B_n$ are adjacent
if and only if $x_i=y_{i+1}$ and $x_{i+1}=y_i$ for some $i$ and
$x_j=y_j$ for all $j\neq i$ and $i+1$. It is obvious $B_n$ is an
$(n-1)$-regular graph.

The crossing number of a graph $G$, denoted by $cr(G)$, is the
smallest number of pairwise crossings of edges among all drawings of
$G$ in the plane. In the past thirty years, it turned out that
crossing number played an important role not only in various fields
of discrete and computational geometry, but also in the design of
VLSI circuits \cite{BL84} and wiring layout problems. However,
computing the crossing number was proved to be NP-complete by Garey
and Johnson \cite{GJ83}. Thus, it is not surprising that the exact
crossing numbers are known for graphs of few families and that the
arguments often strongly depend on their structures (see
\cite{LYLH06,ZLYD08}).

Concerned with the upper bound of the crossing number of $Q_n$,
there is a long-standing conjecture proposed by Erd\H{o}s and Guy
\cite{EG73} in 1973:
$$cr(Q_n)\leq \frac{5}{32}4^n-\lfloor\frac{n^2+1}{2}\rfloor 2^{n-2}.$$
This conjecture remains open till L. Faria, C.M.H de Figueiredo, O.
S\'{y}kora, and I. Vrt'o in \cite{FFSV08} constructed a good drawing
of $Q_n$ which gives the upper bound. However, the upper bound of
the crossing number of $B_n$ is still unknown.

Motivated by \cite{FFSV08}, we give an upper bound of the crossing
number of the bubble-sort graph $B_n$ by constructing a drawing of
$B_n$ in the plane in this article. In Section 2, we introduce some
technical notations and tools, while in Section 3 we give an upper
bound of the crossing number of $B_n$ for $n\geq 5$.

\section{Preliminaries}

\indent \indent We use $V(G)$ and $E(G)$ to denote the vertex set
and the edge set of $G$, respectively. A drawing of $G$ is said to
be a $good$ drawing, provided that no edge crosses itself, no
adjacent edges cross each other, no two edges cross more than once,
and no three edges cross in a point. It is well known that the
crossing number of a graph is attained only in good drawings of the
graph. So we always assume that all drawings throughout this article
are good drawings. For a good drawing $D$ of a graph $G$, we denote
by $\nu(D)$ the number of crossings in $D$. It is clear that
$cr(G)\le\nu(D)$.

Let $A$ and $B$ be two disjoint subsets of an edge set $E$. The
number of the crossings formed by an edge in $A$ and another edge in
$B$ is denoted by $\nu_D(A,B)$ in a drawing $D$. The number of the
crossings that involve a pair of edges in $A$ is denoted by
$\nu_D(A)$. Then $\nu_D(A\cup B)=\nu_D(A)+\nu_D(B)+\nu_D(A,B)$.

Next, we present some topological results. Using the results, we
derive expressions for counting the crossings of the drawing that is
constructed by us in the following section. For $n\geq 6$, let
$P=(k_1,k_2,\cdots,k_{n-2})$ be an arbitrary permutation on
$\{2,3,\cdots,n-1\}$. For $0\leq a\leq n-2$, we define a structure
$mesh$ $M_{n,a}$ depending on $P$ in the real plane $\mathbb{R}^2$
which is used in the counting-crossing process. Join the points
$(0,1)$, $(0,2)$, $(0,3),\cdots,(0,n-1)$ and $(0,n)$ of the vertical
real axis with a line.  For $1\leq i\leq a$, let $S_{l,i}=\{l_{ij}:
j\in \{1,2,\cdots ,n\}-\{k_i\}\}$ be a group of non-vertical
parallel semi-straight lines in the left semi-plane, such that the
point $(0,j)$ belongs to $l_{ij}$ and $l_{11},\cdots, l_{a1}$ lie
anticlockwise around point $(0,1)$. For $a+1\leq i\leq n-2$, let
$S_{r,i}=\{r_{ij}:j\in \{1,2,\cdots ,n\}-\{k_i\}\}$ be a group of
non-vertical parallel semi-straight lines in the right semi-plane,
such that the point $(0,j)$ belongs to $r_{ij}$, and $r_{a+1\
1},\cdots, r_{n-2\ 1}$ lie clockwise around point $(0,1)$. The edge
$l_{i\ k_i}$ ($1\leq i\leq a$) ($r_{i\ k_i}$ ($a+1\leq i\leq n-2$))
 is called the ``lost'' edge with respect to point $(0,k_i)$. In
Figure \ref{fig:(M61&M62)} we show drawings of $M_{6,1}$ and
$M_{6,2}$ depending on permutation $(2,4,5,3)$.

\begin{figure}[ht]
\centering\includegraphics[scale=1]{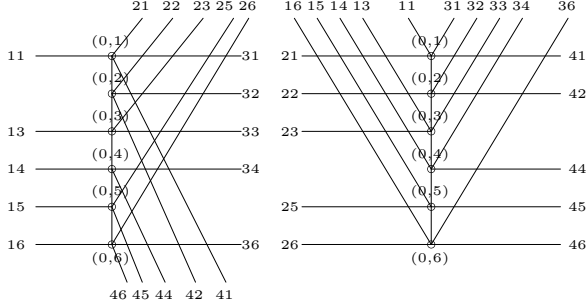}\caption{\small{Drawings
of $M_{6,1}$ and $M_{6,2}$ depending on
$(2,4,5,3)$}}\label{fig:(M61&M62)}
\end{figure}

By counting the crossings in $M_{n,a}$, we can get

\begin{lemma}\label{lemma: sum<max}
For positive integers $n\geq 6$ and $0\leq a\leq n-2$,
$$\begin{array}{rlll} &\nu_D({S_{l,i_1}},S_{l,i_2})\\=&\left
\{\begin{array}{llll}
{n\choose2}-(n-k_{i_2})-(k_{i_1}-1)=\frac{1}{2}n^2-\frac{3}{2}n+1+k_{i_2}-k_{i_1},\ if\  k_{i_1}<k_{i_2}\\
{n\choose2}-(n-k_{i_2})-(k_{i_1}-2)=\frac{1}{2}n^2-\frac{3}{2}n+2+k_{i_2}-k_{i_1},\
if\ k_{i_1}>k_{i_2}
              \end{array}
           \right.\\
\end{array}$$
for $1\leq i_1<i_2\leq a$, and
$$\begin{array}{rlll}
&\nu_D({S_{r,i_1}},S_{r,i_2})\\=&\left \{\begin{array}{llll}
{n\choose2}-(n-k_{i_2})-(k_{i_1}-1)=\frac{1}{2}n^2-\frac{3}{2}n+1+k_{i_2}-k_{i_1},\ if\  k_{i_1}<k_{i_2}\\
{n\choose2}-(n-k_{i_2})-(k_{i_1}-2)=\frac{1}{2}n^2-\frac{3}{2}n+2+k_{i_2}-k_{i_1},\
if\ k_{i_1}>k_{i_2}
              \end{array}
           \right .\\
\end{array}$$
for $a+1\leq i_1<i_2\leq n-2$.

\end{lemma}

For any mesh $M_{n,a}$ depending on $(k_1,k_2,\cdots,k_{n-2})$, let
$M'_{n,a}$ be the mesh depending on $(k'_1,k'_2,\cdots,k'_{n-2})$,
where $\{k'_1,k'_2,\cdots,k'_a\}=\{k_1,k_2,\cdots,$ $k_a\}$,
$\{k'_{a+1},k'_{a+2},\cdots,k'_{n-2}\}=\{k_{a+1},k_{a+2},$
$\cdots,k_{n-2}\}$ and $k'_i<k'_j$ for $1\leq i< j\leq a$ and
$a+1\leq i<j\leq n-2$ separately. Then we have the following lemmas.

\begin{lemma}\label{lemma: sorted}
For positive integers $n\geq 6$ and $0\leq a\leq n-2$,
$\nu_D(E(M_{n,a}))\leq\nu_D(E(M'_{n,a}))$.
\end{lemma}

\begin{proof}
Let $\beta=\beta_l+\beta_r$, where $\beta_l$ (or $\beta_r$) is the
number of inversions in permutation section $(k_1,k_2,\cdots,k_{a})$
(or $(k_{a+1},k_{a+2},\cdots,k_{n-2})$). Let
$\beta'=\beta'_l+\beta'_r$, where $\beta'_l$ (or $\beta'_r$) is the
number of inversions in permutation section
$(k'_1,k'_2,\cdots,k'_{a})$ (or
$(k'_{a+1},k'_{a+2},\cdots,k'_{n-2})$). Then $\beta'=0$. If
$\beta=0$, then
$(k_1,k_2,\cdots,k_{n-2})=(k'_1,k'_2,\cdots,k'_{n-2})$, and
$\nu_D(E(M_{n,a}))=\nu_D(E(M'_{n,a}))$.

Now we consider the case for $\beta=m>0$. Let
$P^0=(k^0_1,k^0_2,\cdots,k^0_{n-2})=(k_1,k_2,\cdots,k_{n-2})$. For
$1\leq t\leq m$, suppose $(k^{t-1}_i,k^{t-1}_{i+1})$ be the first
adjacent inversion in permutation section
$(k^{t-1}_1,k^{t-1}_2,\cdots,k^{t-1}_{a})$ or
$(k^{t-1}_{a+1},k^{t-1}_{a+2},\cdots,$ $k^{t-1}_{n-2})$, let
$P^t=(k^t_1,k^t_2,\cdots,k^t_{n-2})$, where $k^{t}_i=k^{t-1}_{i+1}$,
$k^{t}_{i+1}=k^{t-1}_{i}$ and $k^{t}_j=k^{t-1}_j$ for $1\leq j\leq
i-1$ and $i+2\leq j\leq n-2$.

For $0\leq t\leq m$, Let $\beta^{t}=\beta^{t}_l+\beta^{t}_r$, where
$\beta^{t}_l$ (or $\beta^{t}_r$) is the number of inversions in
permutation section $(k^{t}_1,k^{t}_2,\cdots,k^{t}_{a})$ (or
$(k^{t}_{a+1},k^{t}_{a+2},\cdots,k^{t}_{n-2})$). Then
$\beta^{0}=\beta=m$, $\beta^{t}=\beta^{t-1}-1$ and
$\beta^{m}=0=\beta'$.

Let $M^t_{n,a}$ be the mesh depending on $P^t$. Then
$P^m=(k^m_1,k^m_2,\cdots,k^m_{n-2})$ $=(k'_1,k'_2,\cdots,k'_{n-2})$,
$M^0_{n,a}=M_{n,a}$ and $M^m_{n,a}=M'_{n,a}$. For $1\leq i\leq a$,
let $S^{t}_{l,i}=\{l_{ij}: j\in \{1,2,\cdots ,n\}-\{k^{t}_i\}\}$,
and for $a+1\leq i\leq n-2$ let $S^{t}_{r,i}=\{r_{ij}:j\in
\{1,2,\cdots ,n\}-\{k^{t}_i\}\}$.

By Lemma \ref{lemma: sum<max},
$$\begin{array}{llll}
&\quad\nu_D(E(M^t_{n,a}))-\nu_D(E(M^{t-1}_{n,a}))\\
&=\left
\{\begin{array}{llll}
\nu_{D}(S^{t}_{l,i},S^{t}_{l,i+1})-\nu_{D}(S^{t-1}_{l,i},S^{t-1}_{l,i+1}),\ \ \  if\  1\leq i\leq a-1\\
\nu_{D}(S^{t}_{r,i},S^{t}_{r,i+1})-\nu_{D}(S^{t-1}_{r,i},S^{t-1}_{r,i+1}),\
\  if\ a+1\leq i\leq n-3
              \end{array}
           \right .\\
&=\left(\frac{1}{2}n^2-\frac{3}{2}n+1+k^t_{i+1}-k^t_{i}\right)-\left(\frac{1}{2}n^2-\frac{3}{2}n+2+k^{t-1}_{i+1}-k^{t-1}_{i}\right)\\
&=k^t_{i+1}-k^t_{i}-k^{t-1}_{i+1}+k^{t-1}_{i}-1\\
&=2k^{t-1}_{i}-2k^{t-1}_{i+1}-1>0.
\end{array}$$
Hence $\nu_D(E(M^{t-1}_{n,a}))<\nu_D(E(M^t_{n,a}))$. It follows
$\nu_D(E(M_{n,a}))=\nu_D(E(M^0_{n,a}))$
$<\nu_D(E(M^1_{n,a}))<\cdots<\nu_D(E(M^m_{n,a}))=\nu_D(E(M'_{n,a}))$.
\end{proof}

\begin{lemma}\label{lemma: Xn<Max}
For integers $m\geq 4$ and $a\geq 1$,
$$\begin{array}{llll}
&\quad\nu_D(E(M_{n,a}))\\
&\leq\left \{\begin{array}{llll}
\frac{1}{8}n^4-\frac{25}{24}n^3+3n^2-\frac{23}{6}n+2, \ \ if\ \  n=2m,\ a=m-1,\\
\frac{1}{8}n^4-\frac{25}{24}n^3+\frac{7}{2}n^2-\frac{29}{6}n+2, \ \ if\ \  n=2m,\ a\in\{m-2,m\},\\
\frac{1}{8}n^4-\frac{25}{24}n^3+\frac{25}{8}n^2-\frac{95}{24}n+\frac{7}{4}, if\ \ n=2m-1,\ a\in\{m-2,m-1\}.\\
              \end{array}
           \right .\\
\end{array}
$$
\end{lemma}
\begin{proof}
Firstly, we give a proper permutation, called
$(k''_1,k''_2,\cdots,k''_{n-2})$, depending on which the mesh
structure is denoted by $M''_{n,a}$.
$$
\begin{array}{llll}
&(k''_1,k''_2,k''_3,\cdots,k''_{a})(k''_{a+1},k''_{a+2},k''_{a+3},\cdots,k''_{n-2}) \\
=& \left \{\begin{array}{llll}
(2,4,6,\cdots,n-2)(3,5,7,\cdots,n-1),& if\ \ n=2m,\ a=m-1,\\
(2,3,5,7,\cdots,n-1)(4,6,8,\cdots,n-2),& if\ \ n=2m,\ a=m,\\
(4,6,8,\cdots,n-2)(2,3,5,7,\cdots,n-1),& if\ \ n=2m,\ a=m-2,\\
(2,4,6,\cdots,n-1)(3,5,7,\cdots,n-2),& if\ \ n=2m-1,\ a=m-1,\\
(3,5,7,\cdots,n-2)(2,4,6,\cdots,n-1),& if\ \ n=2m-1,\ a=m-2.\\
              \end{array}
           \right .\\
\end{array}
$$

Then the permutation $(k''_1,k''_2,\cdots,k''_{n-2})$ satisfies
$k''_i<k''_j$ for $1\leq i< j\leq a$ and $a+1\leq i<j\leq n-2$
separately.

We will prove Lemma \ref{lemma: Xn<Max} holds for $(n,a)=(2m,m)$. By
a similar argument, we can  prove Lemma \ref{lemma: Xn<Max} holds
for the other cases of $(n,a)$ and we leave it to readers.

By Lemmas \ref{lemma: sum<max} and \ref{lemma: sorted},
$$\begin{array}{llll}
&\ \ \ \nu_D(E(M_{n,a}))\\&\leq\nu_D(E(M'_{n,a}))\\
&=\sum\limits_{1\leq
i< j\leq a}\nu_{D}({S'_{l,i}},S'_{l,j})+\sum\limits_{a+1\leq i<j\leq n-2}\nu_{D}({S'_{r,i}},S'_{r,j})\\
&=\left(\frac{1}{2}n^2-\frac{3}{2}n+1\right)\left[{a\choose2}+{n-2-a\choose2}\right]+\sum\limits_{1\leq
i< j\leq a}(k'_j-k'_i)+\sum\limits_{a+1\leq i<j\leq n-2}(k'_j-k'_i)\\
\end{array}$$$$\begin{array}{llll}&=\frac{1}{8}(n^4-9n^3+32n^2-48n+24)+\sum\limits^{\frac{n}{2}}_{i=1}(2i-1-\frac{n}{2})k'_i+\sum\limits^{\frac{n}{2}-2}_{i=1}(2i+1-\frac{n}{2})k'_{i+a}\\
&=\frac{1}{8}(n^4-9n^3+32n^2-48n+24)+\sum\limits^{\frac{n}{2}}_{i=1}(2i-1)k'_i+\sum\limits^{\frac{n}{2}-2}_{i=1}(2i+1)k'_{i+a}-\sum\limits^{n-2}_{i=1}\frac{n}{2}k'_{i}\\
&=\frac{1}{8}(n^4-11n^3+34n^2-44n+24)+\sum\limits^{\frac{n}{2}}_{i=1}(2i-1)k'_i+\sum\limits^{\frac{n}{2}-2}_{i=1}(2i+1)k'_{i+a}.
\end{array}$$

Now we will prove
$\sum^{\frac{n}{2}}_{i=1}(2i-1)k'_i+\sum^{\frac{n}{2}-2}_{i=1}(2i+1)k'_{i+a}\leq
\sum^{\frac{n}{2}}_{i=1}(2i-1)k''_i+\sum^{\frac{n}{2}-2}_{i=1}(2i+1)k''_{i+a}$.
We define
$$\theta(i)=\left\{\begin{array}{llll}
2i-1 &if\ \  1\leq i\leq a,\\
2(i-a)+1 &if\ \  a+1\leq i\leq n-2.
\end{array}\right.$$

Let
$P_{0}=(k^{0}_1,k^{0}_2,\cdots,k^{0}_{n-2})=(k'_1,k'_2,\cdots,k'_{n-2})$.
For $1\leq t\leq n-3$, suppose $k^{t-1}_{i_t}=n-t$, let
$j_t=\frac{n}{2}-\frac{t-1}{2}$ for odd $t$, $j_t=n-\frac{t}{2}-1$
for even $t$ and $P_t=(k^t_1,k^t_2,\cdots,k^t_{n-2})$ where
$k^t_{j_t}=k^{t-1}_{i_t}$, $k^t_{i_t}=k^{t-1}_{j_t}$ and
$k^t_s=k^{t-1}_s$ for $1\leq s\leq n-2, s\neq i_t$ and $s\neq j_t$.
Then, for even $t$
\begin{align}
&(k^{t-1}_{\frac{n}{2}-\frac{t}{2}+1},k^{t-1}_{\frac{n}{2}-\frac{t}{2}+2},\cdots,
k^{t-1}_{\frac{n}{2}})(k^{t-1}_{n-\frac{t}{2}},k^{t-1}_{n-\frac{t}{2}+1},\cdots,
k^{t-1}_{n-2})\notag
\\=&(n-t+1,n-t+3,\cdots,n-1)(n-t+2,n-t+4,\cdots,n-2),
\end{align}
and for odd $t$
\begin{align}
&(k^{t-1}_{\frac{n}{2}-\frac{t-1}{2}+1},k^{t-1}_{\frac{n}{2}-\frac{t-1}{2}+2},\cdots,
k^{t-1}_{\frac{n}{2}})(k^{t-1}_{n-\frac{t+1}{2}},k^{t-1}_{n-\frac{t+1}{2}+1},\cdots,
k^{t-1}_{n-2})\notag \\
=&(n-t+2,n-t+4,\cdots,n-1)(n-t+1,n-t+3,\cdots,n-2).
\end{align}
Furthermore, we have $k^t_{i_t}=k^{t-1}_{j_t}\leq n-t=k^{t-1}_{i_t}=
k^t_{j_t}$ and $P_{n-3}=(k''_1,k''_2,\cdots,k''_{n-2})$.

There are two cases depending on $t$.

Case 1. $t$ is even. Then $t-1$ is odd. By expression $(1)$, $1\leq
i_t\leq \frac{n}{2}-\frac{t}{2}+1-1=\frac{n-t}{2}$ or $a+1\leq
i_t\leq n-\frac{t}{2}-1$. Hence,
$$\begin{array}{llll}
\theta(i_t)&\leq\left\{\begin{array}{llll}2\cdot\frac{n-t}{2}-1, &\
\
if\ \ 1\leq i_t\leq\frac{n-t}{2}\\
2\cdot(n-\frac{t}{2}-1-a)+1, &\ \ if\ \ a+1\leq
i_t\leq n-\frac{t}{2}-1\\
\end{array}\right.\\
&= n-t-1\\
&=2(n-\frac{t}{2}-1-a)+1\\
&=\theta(j_t)
\end{array}
$$

Case 2. $t$ is odd. Then $t-1$ is even. By expression $(2)$, $1\leq
i_t\leq \frac{n}{2}-\frac{t-1}{2}+1-1=\frac{n}{2}-\frac{t-1}{2}$ or
$a+1\leq i_t\leq n-\frac{t+1}{2}-1=n-\frac{t+3}{2}$. Hence,
$$\begin{array}{llll}
\theta(i_t)&\leq\left\{\begin{array}{llll}2\cdot(\frac{n}{2}-\frac{t-1}{2})-1,
&\ \
if\ \ 1\leq i_t\leq\frac{n}{2}-\frac{t-1}{2}\\
2\cdot(n-\frac{t+3}{2}-a)+1, &\ \ if\ \ a+1\leq
i_t\leq n-\frac{t+3}{2}\\
\end{array}\right.\\
&\leq n-t\\
&=2(\frac{n}{2}-\frac{t-1}{2})-1\\
&=\theta(j_t)
\end{array}
$$
So $\theta(i_t)\leq\theta(j_t)$ for all $1\leq t\leq n-3$.

For $0\leq t\leq n-3$, let
$f(P_t)=\sum^{\frac{n}{2}}_{i=1}(2i-1)k^t_i+\sum^{\frac{n}{2}-2}_{i=1}(2i+1)k^t_{i+a}$.
Then for $1\leq t\leq n-3$,
$$\begin{array}{llll}
f(P_t)-f(P_{t-1})&=\theta(i_t)k^t_{i_t}+\theta(j_t)k^t_{j_t}-(\theta(i_t)k^{t-1}_{i_t}+\theta(j_t)k^{t-1}_{j_t})\\
&=\theta(i_t)(k^t_{i_t}-k^{t-1}_{i_t})+\theta(j_t)(k^t_{j_t}-k^{t-1}_{j_t})\\
&=(\theta(i_t)-\theta(j_t))(k^t_{i_t}-k^{t}_{j_t})\geq0.
\end{array}$$
So there is $f(P_{t-1})\leq f(P_t)$. It follows $f(P_{0})\leq
f(P_{1})\leq\cdots\leq f(P_{n-3})$. Hence,
$$\begin{array}{llll}
&\ \ \ \nu_D(E(M_{n,a}))\\
&\leq\frac{1}{8}(n^4-11n^3+34n^2-44n+24)+\sum\limits^{\frac{n}{2}}_{i=1}(2i-1)k'_i+\sum\limits^{\frac{n}{2}-2}_{i=1}(2i+1)k'_{i+a}\\
&=\frac{1}{8}(n^4-11n^3+34n^2-44n+24)+f(P_{0})\\
&\leq\frac{1}{8}(n^4-11n^3+34n^2-44n+24)+f(P_{n-3})\\
&=\frac{1}{8}(n^4-11n^3+34n^2-44n+24)+\sum\limits^{\frac{n}{2}}_{i=1}(2i-1)k''_i+\sum\limits^{\frac{n}{2}-2}_{i=1}(2i+1)k''_{i+a}\\
&=\frac{1}{8}n^4-\frac{25}{24}n^3+\frac{7}{2}n^2-\frac{29}{6}n+2.
\end{array}$$\end{proof}
\section{Upper bound for $cr(B_n)$}
\indent \indent In Figure \ref{fig: B2B3B4}, \ref{fig: B5} and
\ref{fig: B6}, we show a drawing $D_2$ ($D_3,D_4,D_5,D_6$) of $B_2$
($B_3,B_4,B_5,B_6$) containing 0 (0, 0, 120, 5196) crossings. So we
can obtain the following theorem:
\begin{figure}[h]
\centering\includegraphics[scale=1]{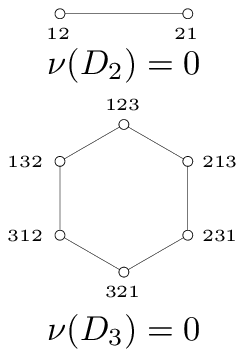}
\centering\includegraphics[scale=1]{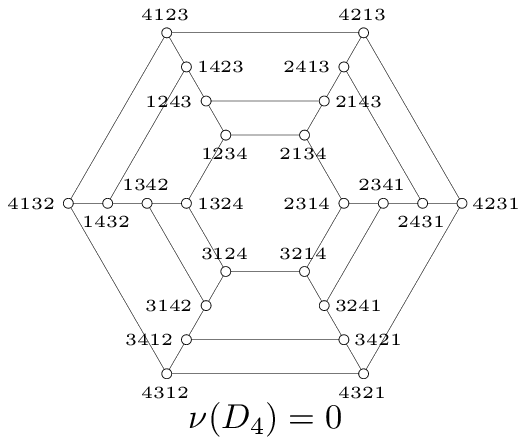}\caption{\small{Some
drawings of $B_2$, $B_3$ and $B_4$}}\label{fig: B2B3B4}
\end{figure}
\begin{figure}[h]
\centering\includegraphics[scale=0.7]{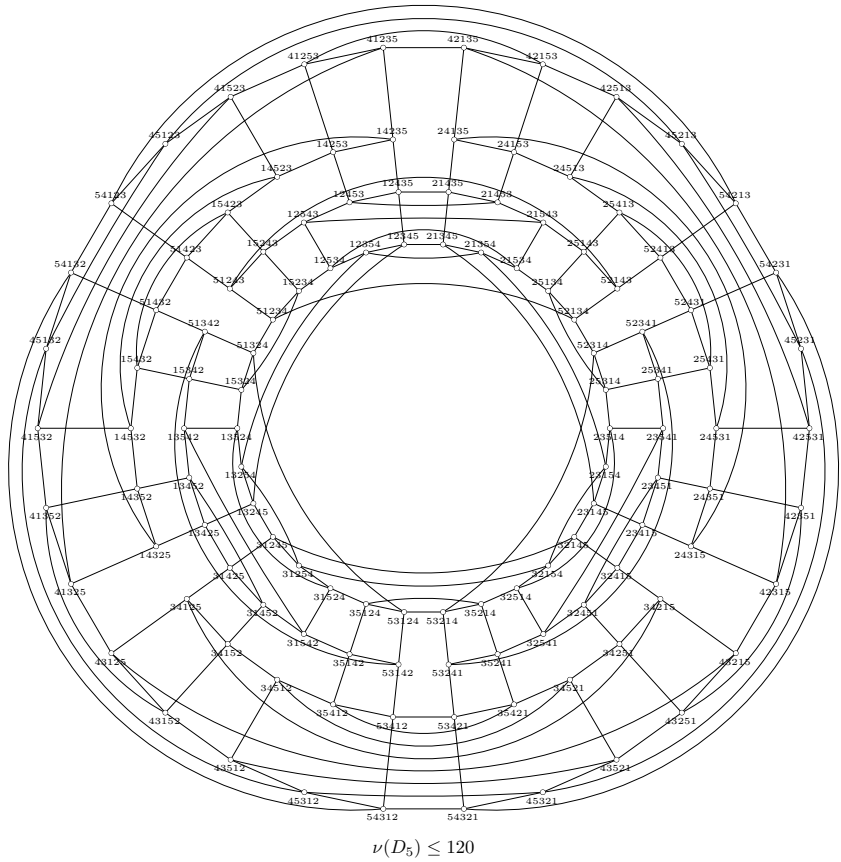}\caption{\small{A
drawing of $B_5$}}\label{fig: B5}
\end{figure}
\begin{figure}[h]
\centering\includegraphics[scale=0.39]{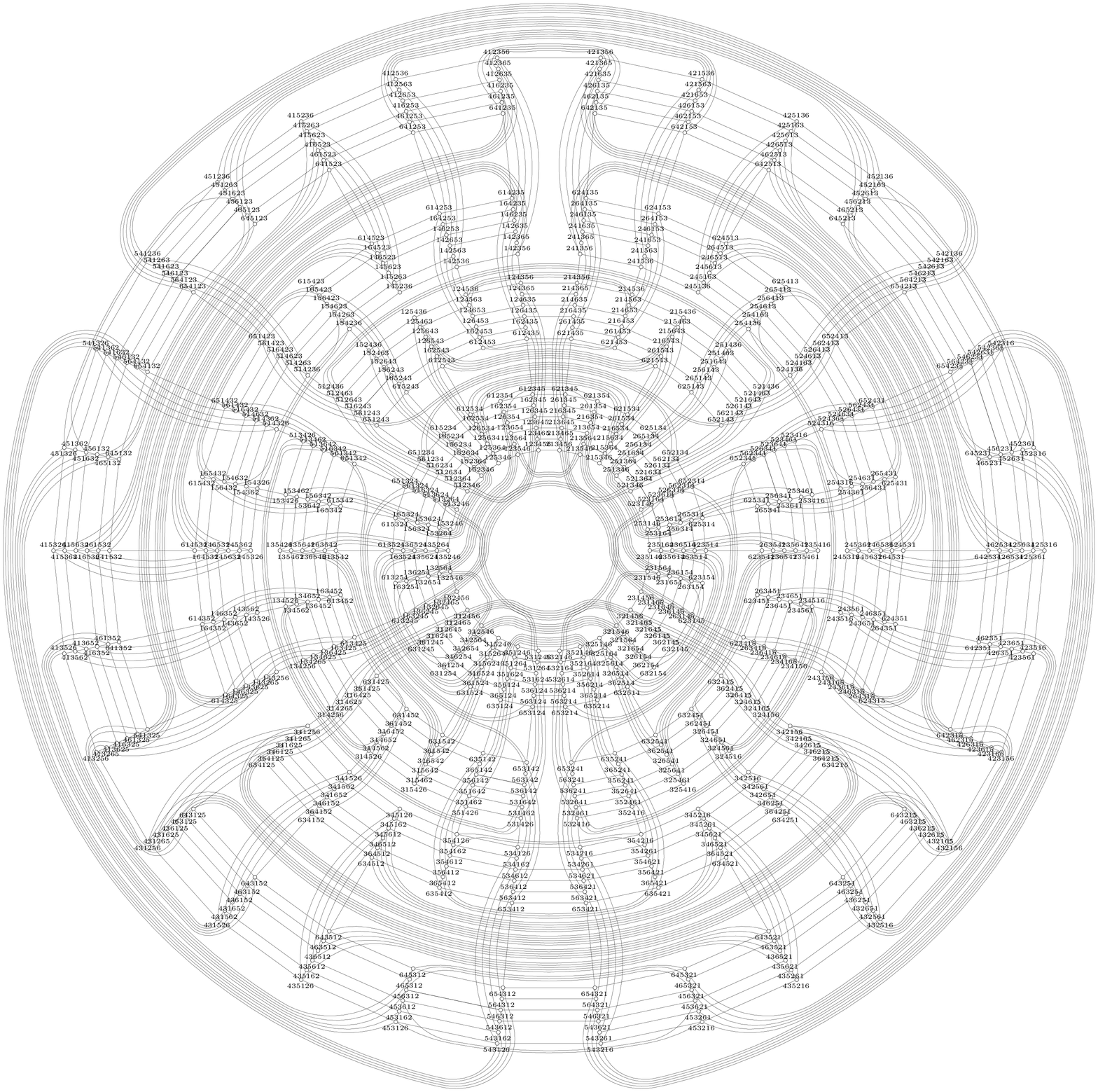}
\centering{\small{$\nu(D_6)\leq 5196$}} \caption{\small{A drawing of
$B_6$}}\label{fig: B6}
\end{figure}

\begin{theorem}\label{theorem: Upper Bound for n=5,6}
$cr(B_2)=cr(B_3)=cr(B_4)=0$, $cr(B_5)\le 120$ and $cr(B_6)\le 5196$.
\end{theorem}

Our main result is Theorem \ref{theorem: main result}.
\begin{theorem}\label{theorem: main result}
For $n\geq7$,
$$\begin{array}{llll}
cr(B_n)\leq&((n-1)!)^2\left[\frac{127}{300}+\frac{5}{24(n-4)!}+\sum\limits^{n-5}_{i=3}\frac{1}{3i!}\right.\\
&\left.+\sum\limits^{n}_{i=2j,j\geq4}\left(\frac{1}{8(i-3)!}-\frac{1}{4(i-2)!}+\frac{1}{4(i-1)!(i-1)}\right)\right].
\end{array}$$
\end{theorem}

To prove Theorem \ref{theorem: main result}, we first need some
definitions and observations.

\begin{definition} For any vertex $v\in V(B_n)$, we denote
$v^{1}$, $v^{2},\cdots,$ $v^{n+1}$ as $n+1$ vertices of $B_{n+1}$
satisfying  $v^{i}=v_1v_2\cdots v_{i-1}(n+1)v_{i}\cdots v_n$ $(1\leq
i\leq n+1)$.
\end{definition}

\begin{observation} For any vertex $v\in V(B_n)$  and $v^{i}\in
V(B_{n+1})$ $ (1\leq i\leq n+1)$, there are $n$ edges \{$v^{1}v^{2},
v^{2}v^{3},\cdots,$ $v^{n}v^{n+1}$\} resulting in a path with $n+1$
vertices.
\end{observation}

\begin{definition} For $n\geq6$, let $B'_n$ be a subgraph of $B_n$ satisfying :
$(1)$ $V(B'_n)$ is a subset of $V(B_n)$ and $V(B'_n)$=$\{v:\mbox{the
label of }v\mbox{ contains }(1,2,3,$ $4),(1,2,4,3),(1,4,2,3)\mbox{
or }(4,1,2,3)$ as a subpermutation$\}$ $(2)$ $uv \in E(B'_n)$,
\textbf{iff} $uv \in E(B_n)$ and at least one of u and v is in
$V(B'_n)$. A drawing of $B'_n$ is denoted by $D'_n$.
\end{definition}

Let $D_n$ be a drawing of $B_n$ containing six $D'_n$, and we have
\begin{observation}\label{Observation relation}
By symmetry $($see Figure $\ref{fig: B6}$ and $\ref{fig: D'6})$,
$\nu(D_n)=6*\nu(D'_n)$.
\end{observation}

In Figure \ref{fig: D'6}, we show a drawing $D'_6$. Based on the
drawing $D'_6$, we construct 5 real number axes $R_i$ ($1\leq
i\leq5$) as shown in Figure \ref{fig: D'6}. All of the axes are
parallel and all the vertices of $B'_n$ are drawn precisely on
$R_i$.
\begin{figure}[h]
\centering\includegraphics[scale=0.55]{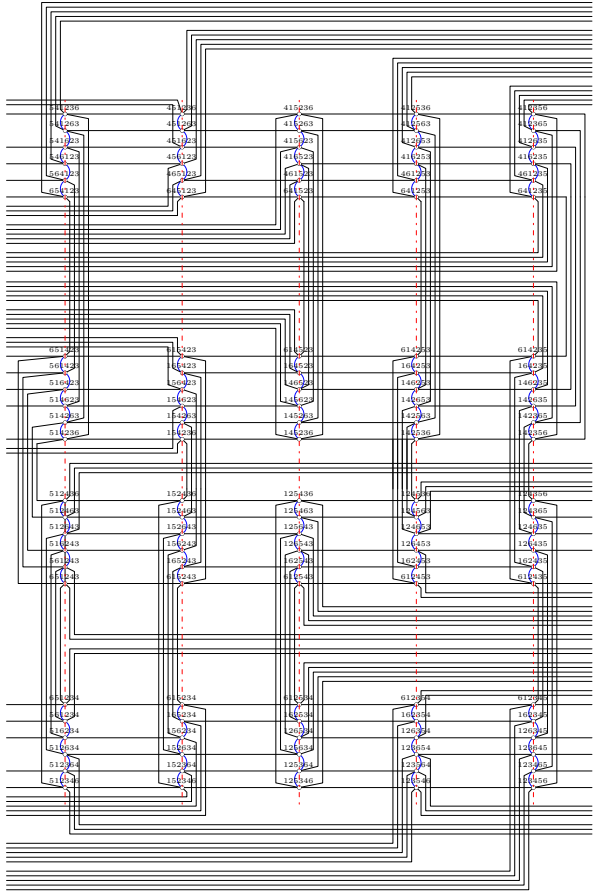}\caption{\small{A
drawing $D'_6$}}\label{fig: D'6}
\end{figure}

\begin{definition}
Let $x$ be a vertex of $B'_n$ on axis $R_i$ and $e=xy$ an edge of
$B'_n$ in $D'_n$. $(1)$ If the edge $e$ shots away the axis $R_i$
left from the vertex $x$, i.e. the part of $e$ in a small enough
neighborhood of $x$ lies in the left of the axis $R_i$, we call $xy$
an \textbf{`l-arc'} with respect to $x$. $(2)$ If $e$ shots away the
axis $R_i$ right from the vertex $x$, we call $xy$ an
\textbf{`r-arc'} with respect to $x$. $(3)$ We write $l(x)=|\{xy\in
E(B'_n):xy \mbox{ is an l-arc with respect to }x\}|$ and
$r(x)=|\{xy\in E(B'_n):xy \mbox{ is an r-arc with respect}$ to
$x\}|$.
\end{definition}

Next we will construct $D'_{n+1}$ from $D'_n$ ($n\geq 6$) and
compute the crossing number $\nu(D'_{n+1})$.

For $1\leq i\leq 5$, let $v$ be the $j$-th vertex lying on $R_i$
(from the top down). We replace $v$ of $B'_n$ in the ``small''
neighborhood of $v$ in the drawing $D'_n$ by path
$P_{v^{n+1}v^{n}\cdots v^{1}}$ ($P_{v^1v^2\cdots v^{n+1}}$) on $R_i$
for odd $j$ (even $j$) (See Figure \ref{fig: D'6}). Now every drawn
edge $e$ in $D'_n$ which started in $v$ will be replaced by $n$
``parallel'' edges (bunch) and drawn along the original edge $e$.
Notice that, in this case, locally we have a drawing of mesh. Doing
this carefully we get a drawing $D'_{n+1}$ for $n\geq 6$.

To construct $D'_{n+1}$ from $D'_n$, we use three kinds of mesh
structures defined as follows:

\begin{definition}
A mesh $M_{n,a}$ $(0\leq a\leq n-2)$ is called a
\textbf{11-structure} mesh, if for $1\leq j\leq
\lfloor\frac{n}{2}\rfloor$, all edges $v^{2j-1}v^{2j}$ are drawn on
the left of $R_i$, and all edges $v^{2j}v^{2j+1}$ are drawn on the
right of $R_i$ $(1\leq i\leq 5)$ $($See Figure $\ref{fig:
11-structure})$.
\end{definition}

\begin{figure}[ht]
\centering\includegraphics[scale=0.8]{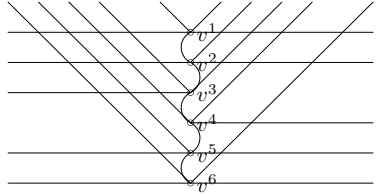}\caption{\small{An
example of 11-structure mesh for $n=6$}} \label{fig: 11-structure}
\end{figure}

\begin{definition}
A mesh $M_{n,a}$ $(0\leq a\leq n-2)$ is called a
\textbf{20-structure} $(\textbf{02-structure})$ mesh, if $n$ is even
and for $1\leq j\leq n-1$, all edges $v^jv^{j+1}$ are drawn on the
left $(right)$ of $R_i$ $(1\leq i\leq 5)$ $($See Figure $\ref{fig:
20-structure&02-structure})$.
\end{definition}
\begin{figure}[ht]
\centering\includegraphics[scale=0.9]{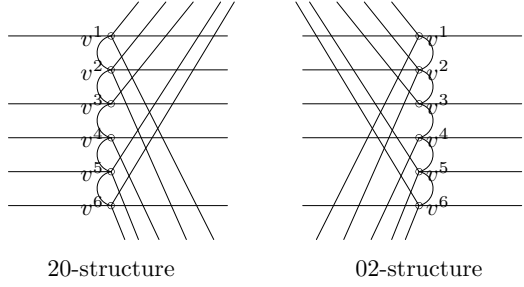}\caption{\small{Examples
of 20-structure mesh and 02-structure mesh for $n=6$}}\label{fig:
20-structure&02-structure}
\end{figure}

Now we introduce a $balanced$ drawing combined with 11-structure,
20-structure and 02-structure meshes. For $n\geq 6$, every vertex
$v\in V(B'_{n})$ will be replaced by the path $P_{v^{n+1}v^{n}\cdots
v^{1}}$ (or $P_{v^1v^2\cdots v^{n+1}}$) of a specific mesh
$M_{n+1,l(v)}(v)$ depending on $l(v)-r(v)$ as shown in Table
\ref{Table: balanced drawing}. We assume that all drawings $D'_n$
($n\geq6$) in the following part are balanced drawings. A part of
$D'_8$ containing 11-structure, 20-structure and 02-structure meshes
is shown in Figure \ref{fig: part of D'8}.

\begin{table}[ht]
\captionstyle{center}
 \caption{ Every vertex in $B'_n$ is replaced by \protect \\ the $(n+1)$-path of a specific
mesh}\label{Table: balanced drawing}
\centering
\begin{tabular}{c c
c} \hline
        $ n$ &  $l(v)-r(v)$& $specific\ mesh$\\
\hline
    $even$& $-1,1$&$11-structure$\\
    $odd$& $0$&$11-structure$\\
    $odd$& $-2$&$20-structure$\\
    $odd$& $2$&$02-structure$\\
\hline
\end{tabular}
\end{table}

\begin{figure}[ht]
\centering\includegraphics[scale=0.8]{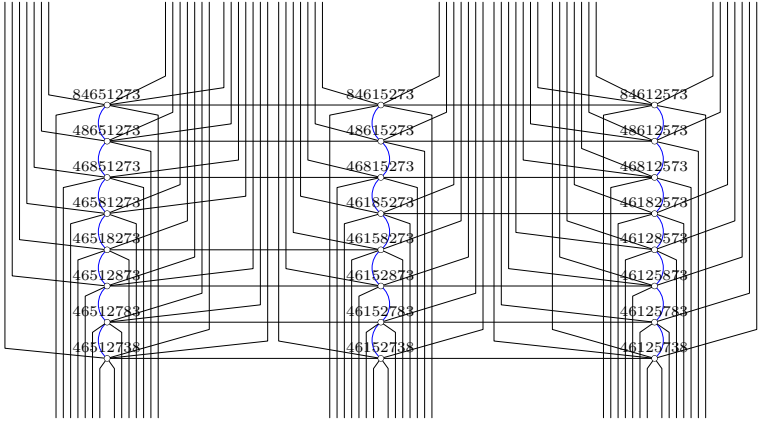}
\captionstyle{center}\caption{\small{A part of $D'_8$ containing
11-structure, 20-structure \protect \\ and 02-structure
meshes}}\label{fig: part of D'8}
\end{figure}

\begin{lemma}\label{lemma: |l-s|<=2}
For any vertex $v$ in the drawing $D'_{n}$ with $n\geq 6$,
$l(v)-r(v)\in \{0,2,-2\}$ for odd $n$, and $l(v)-r(v)\in\{1,-1\}$
for even $n$.
\end{lemma}

\begin{proof}
We prove Lemma \ref{lemma: |l-s|<=2} by induction on n.

1) For $n=6$, there is $l(v)-r(v)\in\{1,-1\}$ for any vertex $v$ in
the drawing $D'_{n}$ (See Figure \ref{fig: D'6}). Hence Lemma
\ref{lemma: |l-s|<=2} holds for $n=6$.

2) Assume Lemma \ref{lemma: |l-s|<=2} holds for $k$. Now consider
the case for $k+1$. There are two cases depending on whether $k$ is
even or odd.

Case 1. $k$ is even. $v$ is replaced by a 11-structure mesh.

For $2\leq i\leq k$, considering the lost edge with respect to
$v^i$, we have $l(v)-1+1\leq l(v^i)\leq l(v)+1$ and $r(v)-1+1\leq
r(v^i)\leq r(v)+1$. For $i=1$ or $k+1$, since there is no lost edge
with respect to $v^i$, we have $l(v)\leq l(v^i)\leq l(v)+1$ and
$r(v)\leq r(v^i)\leq r(v)+1$. It follows, $|l(v^i)-r(v^i)|\leq
|l(v)-r(v)|+1=2$ for all $1 \leq i\leq k+1$. Since $l(v^i)+r(v^i)=k$
is even, we have $l(v^i)-r(v^i)\in \{0,-2,2\}$.

Case 2. $k$ is odd.

Case 2.1. $l(v)-r(v)=0$. $v$ is replaced by a 11-structure mesh.

For $2\leq i\leq k$, we have $l(v)-1+1\leq l(v^i)\leq l(v)+1$ and
$r(v)-1+1\leq r(v^i)\leq r(v)+1$. For $i=1$ or $k+1$, we have
$l(v)\leq l(v^i)\leq l(v)+1$ and $r(v)\leq r(v^i)\leq r(v)+1$. It
follows $|l(v^i)-r(v^i)|\leq |l(v)-r(v)|+1=1$ for all $1 \leq i\leq
k+1$.

Case 2.2. $l(v)-r(v)=2$. $v$ is replaced by a 02-structure mesh.

For $2\leq i\leq k$, we have $l(v)-1\leq l(v^i)\leq l(v)$ and
$r(v)-1+2\leq r(v^i)\leq r(v)+2$. It follows $-1=l(v)-1-r(v)-2\leq
l(v^i)-r(v^i)\leq l(v)-r(v)-1=1$. For $i=1$ or $k+1$, we have
$l(v^i)= l(v)$ and $r(v^i)=r(v)+1$. Then $l(v^i)-r(v^i)=
l(v)-r(v)-1=1$. It follows $|l(v^i)-r(v^i)|\leq 1$ for all $1 \leq
i\leq k+1$.

Case 2.3. $l(v)-r(v)=-2$. $v$ is replaced by a 20-structure mesh.

For $2\leq i\leq k$, we have $l(v)-1+2\leq l(v^i)\leq l(v)+2$ and
$r(v)-1\leq r(v^i)\leq r(v)$. It follows $-1=l(v)+1-r(v)\leq
l(v^i)-r(v^i)\leq l(v)+2-r(v)+1=1$. For $i=1$ or $k+1$, we have
$l(v^i)= l(v)+1$ and $r(v^i)=r(v)$. Then $l(v^i)-r(v^i)=
l(v)+1-r(v)=-1$. It follows $|l(v^i)-r(v^i)|\leq 1$ for all $1 \leq
i\leq k+1$.

By Cases 2.1-2.3, we have $|l(v^i)-r(v^i)|\leq 1$ for odd $k$. Since
$l(v^i)+r(v^i)=k$ is odd, we have $l(v^i)-r(v^i)\in\{-1,1\}$.

By 1) and 2), Lemma \ref{lemma: |l-s|<=2} holds for all $n\geq6$.
\end{proof}

\begin{lemma}\label{lemma: numbers in odd}
For even $n\geq 6$ and any vertex $v$ in the drawing $D'_{n}$, there
is $|\{v^i:l(v^i)-r(v^i)\in\{2,-2\}\}|=\frac{n}{2}$ and
$|\{v^i:l(v^i)-r(v^i)=0\}|=\frac{n}{2}+1$ in $D'_{n+1}$.
\end{lemma}

\begin{proof}By Lemma
\ref{lemma: |l-s|<=2}, $l(v)-r(v)\in\{-1,1\}$ for even $n\geq6$.

Case 1. $(l(v),r(v))=(\frac{n}{2}-1,\frac{n}{2})$.
$$\begin{array}{llll}
&|\{v^i:l(v^i)-r(v^i)\in\{2,-2\}\}|\\
=&|\{v^i:l(v^i)-r(v^i)=-2\}|\\
=&|\{v^i:(l(v^i),r(v^i))=(\frac{n}{2}-1-1+1,\frac{n}{2}+1)\}\cup\\
&\{v^{n+1}:(l(v^{n+1}),r(v^{n+1}))=(\frac{n}{2}-1,\frac{n}{2}+1)\}|\\
=&|\{v^i:(l(v^i),r(v^i))=(\frac{n}{2}-1,\frac{n}{2}+1)\}|+\\
&|\{v^{n+1}:(l(v^{n+1}),r(v^{n+1}))=(\frac{n}{2}-1,\frac{n}{2}+1)\}|\\
=&l(v)+1=\frac{n}{2}-1+1=\frac{n}{2}.\\ \\
&|\{v^i:l(v^i)-r(v^i)=0\}| \hspace{500bp}\\
=&|\{v^i:(l(v^i),r(v^i))=(\frac{n}{2}-1+1,\frac{n}{2}-1+1)\}\cup \\
&\{v^{1}:(l(v^{1}),r(v^{1}))=(\frac{n}{2}-1+1,\frac{n}{2})\}|\\
=&|\{v^i:(l(v^i),r(v^i))=(\frac{n}{2},\frac{n}{2})\}|+|\{v^{1}:(l(v^{1}),r(v^{1}))=(\frac{n}{2},\frac{n}{2})\}|\\
=&r(v)+1=\frac{n}{2}+1.
\end{array}$$
Case 2. $(l(v),r(v))=(\frac{n}{2},\frac{n}{2}-1)$.
$$\begin{array}{llll}
&|\{v^i:l(v^i)-r(v^i)=0\}|\hspace{500bp}\\
=&|\{v^i:(l(v^i),r(v^i))=(\frac{n}{2}-1+1,\frac{n}{2}-1+1)\}\cup\\
&\{v^{n+1}:(l(v^{n+1}),r(v^{n+1}))=(\frac{n}{2},\frac{n}{2}-1+1)\}|\\
=&|\{v^i:(l(v^i),r(v^i))=(\frac{n}{2},\frac{n}{2})\}|+|\{v^{n+1}:(l(v^{n+1}),r(v^{n+1}))=(\frac{n}{2},\frac{n}{2})\}|\\
=&l(v)+1=\frac{n}{2}+1.
\end{array}$$
$$\begin{array}{llll}
&|\{v^i:l(v^i)-r(v^i)\in\{2,-2\}\}|\hspace{500bp}\\
=&|\{v^i:l(v^i)-r(v^i)=2\}|\\
=&|\{v^i:(l(v^i),r(v^i))=(\frac{n}{2}+1,\frac{n}{2}-1-1+1)\}\cup\\
&\{v^{1}:(l(v^{1}),r(v^{1}))=(\frac{n}{2}+1,\frac{n}{2}-1)\}|\\
=&|\{v^i:(l(v^i),r(v^i))=(\frac{n}{2}+1,\frac{n}{2}-1)\}|+\\
&|\{v^{1}:(l(v^{1}),r(v^{1}))=(\frac{n}{2}+1,\frac{n}{2}-1)\}|\\
=&r(v)+1=\frac{n}{2}-1+1=\frac{n}{2}.
\end{array}$$
\end{proof}

To construct $D'_{n+1}$ from $D'_{n}$, it can be obtained that there
are totally 2 types of crossings in $D'_{n+1}$.

\noindent\textbf{Type 1. } \textbf{Crossings produced by the
crossings in $D'_n$}. As shown in Figure \ref{fig: type1}, each
crossing in $D'_n$ will be replaced by $n^2$ crossings in
$D'_{n+1}$.

\begin{figure}[ht]
\centering
\includegraphics[scale=0.8]{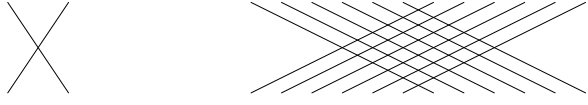}\caption{\small{An example of the
crossings of Type 1} for $n=6$}\label{fig: type1}
\end{figure}

\noindent\textbf{Type 2. } \textbf{Crossings produced by replacing a
vertex of $D'_{n}$ by $n+1$ vertices in $D'_{n+1}$ (See Figure
\ref{fig:(M61&M62)}).}

Then $\nu(D'_{n+1})$ will be $n^2\nu(D'_n)$ plus the sum of the
numbers of the crossings in the small neighborhood of the
$M_{n+1,l(v)}(v)$ for all vertices $v\in V(B'_n)$.

\begin{lemma}\label{lemma: n>=7 odd}
For $n=2m-2$ $(m\geq 4)$,$$ \nu(D'_{n+1})\leq
n^2\nu(D'_{n})+\frac{n!}{144}(3n^4-13n^3+18n^2-8n).$$
\end{lemma}

\begin{proof} By Lemma \ref{lemma: |l-s|<=2}, $l(v)-r(v)\in\{-1,1\}$ for any
vertex $v$ in $D'_{n}$. It follows $l(v)\in\{m-2,m-1\}$, by Lemma
\ref{lemma: Xn<Max}, we have
$$\begin{array}{llll}
&\quad\nu(D'_{n+1})\\&\leq
n^2\nu(D'_{n})+\frac{n!}{6}\cdot\left[\frac{1}{8}(n+1)^4-\frac{25}{24}(n+1)^3+\frac{25}{8}(n+1)^2-\frac{95}{24}(n+1)+\frac{7}{4}\right]\\
&=n^2\nu(D'_{n})+\frac{n!}{144}(3n^4-13n^3+18n^2-8n)
\end{array}$$ for $n=2m-2$ and $m\geq 4$.
\end{proof}

\begin{lemma}\label{lemma: n>=7 even}
For $n=2m-1$ $(m\geq 4)$,$$ \nu(D'_{n+1})\leq
n^2\nu(D'_{n})+\frac{(n-1)!}{144}(3n^5-13n^4+21n^3-17n^2+6).$$
\end{lemma}
\begin{proof} By Lemma \ref{lemma: |l-s|<=2}, $l(v)-r(v)\in\{0,-2,2\}$ for any
vertex $v$ in $D'_{n}$. It follows $l(v)\in\{m,m-2,m-1\}$, by Lemmas
\ref{lemma: Xn<Max} and \ref{lemma: numbers in odd}, we have
$$\begin{array}{llll}
&\quad\nu(D'_{n+1})\\
&\leq n^2\nu(D'_{n})+\frac{n!}{6}\cdot\frac{1}{n}\left\{\frac{n+1}{2}\left[\frac{1}{8}(n+1)^4-\frac{25}{24}(n+1)^3+3(n+1)^2-\frac{23}{6}(n+1)\right.\right.\\
&\quad+2]+\left.\frac{n-1}{2}\left[\frac{1}{8}(n+1)^4-\frac{25}{24}(n+1)^3+\frac{7}{2}(n+1)^2-\frac{29}{6}(n+1)+2\right]\right\}\\
&=n^2\nu(D'_{n})+\frac{(n-1)!}{144}(3n^5-13n^4+21n^3-17n^2+6)
\end{array}$$
for $n=2m-1$ and $m\geq4$.
\end{proof}

By Observations \ref{Observation relation} and Lemmas \ref{lemma:
n>=7 odd} and \ref{lemma: n>=7 even}, we can induce a general
expression for $n\geq7$.
\begin{lemma}\label{lemma: n>=7}
For $n\geq7$,$$\begin{array}{llll} &\quad\nu(D_{n})\\&\leq\left(\frac{(n-1)!}{5!}\right)^2\nu(D_6)\\
&\quad+\sum\limits^{n}_{i=2j+1,j\geq3}\left(\left(\frac{(n-1)!}{(i-1)!}\right)^2\cdot\frac{(i-1)!}{24}(3i^4-25i^3+75i^2-95i+42)\right)\\
&\quad+\sum\limits^{n}_{i=2j,j\geq4}\left(\left(\frac{(n-1)!}{(i-1)!}\right)^2\cdot\frac{(i-2)!}{24}(3i^5-28i^4+103i^3-188i^2+164i-48)\right).
\end{array}$$
\end{lemma}
\begin{proof} By Observation \ref{Observation relation} and Lemma
\ref{lemma: n>=7 odd}, we have a recursive expression for $n=2m-1$
and $m\geq 4$
$$\begin{array}{llll}
&\quad\nu(D_n)\\&=6\nu(D'_n)\\
&\leq6\Big\{(n-1)^2\nu(D'_{n-1})\\
&\quad\left.+\frac{(n-1)!}{144}[3(n-1)^4-13(n-1)^3+18(n-1)^2-8(n-1)]\right\}\\
&=(n-1)^2\nu(D_{n-1})+\frac{(n-1)!}{24}\left(3n^4-25n^3+75n^2-95n+42\right)\\
&=(n-1)^2\nu(D_{n-1})+\sum\limits^{n}_{i=n}\left(\left(\frac{(n-1)!}{(i-1)!}\right)^2\cdot\frac{(i-1)!}{24}(3i^4-25i^3+75i^2-95i+42)\right).
\end{array}$$

By Observation \ref{Observation relation} and Lemma \ref{lemma: n>=7
even}, we have a recursive expression for $n=2m$ and $m\geq 4$
$$\begin{array}{llll}
&\quad\nu(D_n)\\&=6\nu(D'_n)\\
&\leq6\Big\{(n-1)^2\nu(D'_{n-1})\\
&\quad\left.+\frac{(n-2)!}{144}[3(n-1)^5-13(n-1)^4+21(n-1)^3-17(n-1)^2+6]\right\}\\
&=(n-1)^2\nu(D_{n-1})+\frac{(n-2)!}{24}\left(3n^5-28n^4+103n^3-188n^2+164n-48\right)\\
&=(n-1)^2\nu(D_{n-1})\\
&\quad+\sum\limits^{n}_{i=n}\left(\left(\frac{(n-1)!}{(i-1)!}\right)^2\cdot\frac{(i-2)!}{24}(3i^5-28i^4+103i^3-188i^2+164i-48)\right).
\end{array}$$

Then we proceed by induction on $n$.

1) The base case is $n=7$.
$$\begin{array}{llll}
&\quad\nu(D_{7})\\
&\leq6^2\nu(D_{6})+\sum\limits^{7}_{i=7}\left(\left(\frac{(n-1)!}{(i-1)!}\right)^2\cdot\frac{(i-1)!}{24}(3i^4-25i^3+75i^2-95i+42)\right)\\
&=\left(\frac{(n-1)!}{5!}\right)^2\nu(D_6)\\
&\quad+\sum\limits^{n}_{i=2j+1,j\geq3}\left(\left(\frac{(n-1)!}{(i-1)!}\right)^2\cdot\frac{(i-1)!}{24}(3i^4-25i^3+75i^2-95i+42)\right)\\
&\quad+\sum\limits^{n}_{i=2j,j\geq4}\left(\left(\frac{(n-1)!}{(i-1)!}\right)^2\cdot\frac{(i-2)!}{24}(3i^5-28i^4+103i^3-188i^2+164i-48)\right).
\end{array}$$

2) We assume this lemma is true for $n=k-1$.

For $n=k=2m$ and $m\geq4$, we have
$$\begin{array}{llll}
&\quad\nu(D_{k})\\
&\leq (k-1)^2\nu(D_{k-1})\\
&\quad+\sum\limits^{k}_{i=k}\left(\left(\frac{(k-1)!}{(i-1)!}\right)^2\cdot\frac{(i-2)!}{24}(3i^5-28i^4+103i^3-188i^2+164i-48)\right)\\
&= (k-1)^2\left[\left(\frac{(k-2)!}{5!}\right)^2\nu(D_6)\right.\\
&\quad+\sum\limits^{k-1}_{i=2j+1,j\geq3}\left(\left(\frac{(k-2)!}{(i-1)!}\right)^2\cdot\frac{(i-1)!}{24}(3i^4-25i^3+75i^2-95i+42)\right)\\
&\quad\left.+\sum\limits^{k-1}_{i=2j,j\geq4}\left(\left(\frac{(k-2)!}{(i-1)!}\right)^2\cdot\frac{(i-2)!}{24}(3i^5-28i^4+103i^3-188i^2+164i-48)\right)\right]\\
&\quad+\sum\limits^{k}_{i=k}\left(\left(\frac{(k-1)!}{(i-1)!}\right)^2\cdot\frac{(i-2)!}{24}(3i^5-28i^4+103i^3-188i^2+164i-48)\right)\\
&= \left(\frac{(k-1)!}{5!}\right)^2\nu(D_6)\\
&\quad+\sum\limits^{k}_{i=2j+1,j\geq3}\left(\left(\frac{(k-1)!}{(i-1)!}\right)^2\cdot\frac{(i-1)!}{24}(3i^4-25i^3+75i^2-95i+42)\right)\\
&\quad+\sum\limits^{k}_{i=2j,j\geq4}\left(\left(\frac{(k-1)!}{(i-1)!}\right)^2\cdot\frac{(i-2)!}{24}(3i^5-28i^4+103i^3-188i^2+164i-48)\right).
\end{array}$$
For $n=k=2m+1$ and $m\geq4$, we have
$$\begin{array}{llll}
&\quad\nu(D_{k})\\
&\leq (k-1)^2\nu(D_{k-1})+\sum\limits^{k}_{i=k}\left(\left(\frac{(k-1)!}{(i-1)!}\right)^2\cdot\frac{(i-1)!}{24}(3i^4-25i^3+75i^2-95i+42)\right)\\
\end{array}$$
$$\begin{array}{llll}&= (k-1)^2\left[\left(\frac{(k-2)!}{5!}\right)^2\nu(D_6)\right.\\
&\quad+\sum\limits^{k-1}_{i=2j+1,j\geq3}\left(\left(\frac{(k-2)!}{(i-1)!}\right)^2\cdot\frac{(i-1)!}{24}(3i^4-25i^3+75i^2-95i+42)\right)\\
&\quad\left.+\sum\limits^{k-1}_{i=2j,j\geq4}\left(\left(\frac{(k-2)!}{(i-1)!}\right)^2\cdot\frac{(i-2)!}{24}(3i^5-28i^4+103i^3-188i^2+164i-48)\right)\right]\\
&\quad+\sum\limits^{k}_{i=k}\left(\left(\frac{(k-1)!}{(i-1)!}\right)^2\cdot\frac{(i-1)!}{24}(3i^4-25i^3+75i^2-95i+42)\right)\\
&= \left(\frac{(k-1)!}{5!}\right)^2\nu(D_6)\\
&\quad+\sum\limits^{k}_{i=2j+1,j\geq3}\left(\left(\frac{(k-1)!}{(i-1)!}\right)^2\cdot\frac{(i-1)!}{24}(3i^4-25i^3+75i^2-95i+42)\right)\\
&\quad+\sum\limits^{k}_{i=2j,j\geq4}\left(\left(\frac{(k-1)!}{(i-1)!}\right)^2\cdot\frac{(i-2)!}{24}(3i^5-28i^4+103i^3-188i^2+164i-48)\right).
\end{array}$$
By 1) and 2), Lemma \ref{lemma: n>=7} holds for all $n\geq7$.
\end{proof}

By Lemma \ref{lemma: n>=7}, for $n\geq7$ we have
$$\begin{array}{llll}
&\quad cr(B_n)\\
&\leq\nu(D_n)\\
&\leq\left(\frac{(n-1)!}{5!}\right)^2\nu(D_6)\\
&\quad+\sum\limits^{n}_{i=2j+1,j\geq3}\left(\left(\frac{(n-1)!}{(i-1)!}\right)^2\cdot\frac{(i-1)!}{24}(3i^4-25i^3+75i^2-95i+42)\right)\\
&\quad+\sum\limits^{n}_{i=2j,j\geq4}\left(\left(\frac{(n-1)!}{(i-1)!}\right)^2\cdot\frac{(i-2)!}{24}(3i^5-28i^4+103i^3-188i^2+164i-48)\right)\\
&=\frac{433}{1200}((n-1)!)^2+\sum\limits^{n}_{i=2j+1,j\geq3}\left[\frac{((n-1)!)^2}{24(i-1)!}(3i^4-25i^3+75i^2-95i+42)\right]\\
&\quad+\sum\limits^{n}_{i=2j,j\geq4}\left[\frac{((n-1)!)^2}{24(i-1)!(i-1)}(3i^5-28i^4+103i^3-188i^2+164i-48)\right]\\
&=\frac{433}{1200}((n-1)!)^2+\sum\limits^{n}_{i=2j+1,j\geq3}\frac{((n-1)!)^2}{24}\cdot\left[\frac{3}{(i-5)!}+\frac{5}{(i-4)!}\right]\\
&\quad+\sum\limits^{n}_{i=2j,j\geq4}\frac{((n-1)!)^2}{24}\cdot\left[\frac{3}{(i-5)!}+\frac{5}{(i-4)!}+\frac{3}{(i-3)!}-\frac{6}{(i-2)!}+\frac{6}{(i-1)!(i-1)}\right]\\
&=\frac{433}{1200}((n-1)!)^2+\frac{((n-1)!)^2}{24}\cdot\left[\sum\limits^{n}_{i=7}\left(\frac{3}{(i-5)!}+\frac{5}{(i-4)!}\right)\right.\\
&\quad+\sum\limits^{n}_{i=2j,j\geq4}\left(\frac{3}{(i-3)!}-\frac{6}{(i-2)!}+\frac{6}{(i-1)!(i-1)}\right)\bigg]\\
\end{array}$$
$$\begin{array}{llll}&=\frac{433}{1200}((n-1)!)^2+\frac{((n-1)!)^2}{24}\cdot\left[\left(\frac{3}{2}+\sum\limits^{n-5}_{i=3}\frac{8}{i!}+\frac{5}{(n-4)!}\right)\right.\\
&\quad+\sum\limits^{n}_{i=2j,j\geq4}\left(\frac{3}{(i-3)!}-\frac{6}{(i-2)!}+\frac{6}{(i-1)!(i-1)}\right)\bigg]\\
&=\frac{127}{300}((n-1)!)^2+\frac{((n-1)!)^2}{24}\cdot\left[\frac{5}{(n-4)!}+\sum\limits^{n-5}_{i=3}\frac{8}{i!}\right.\\
&\quad+\sum\limits^{n}_{i=2j,j\geq4}\left(\frac{3}{(i-3)!}-\frac{6}{(i-2)!}+\frac{6}{(i-1)!(i-1)}\right)\bigg]\\
&=((n-1)!)^2\left[\frac{127}{300}+\frac{5}{24(n-4)!}+\sum\limits^{n-5}_{i=3}\frac{1}{3i!}\right.\\
&\quad+\sum\limits^{n}_{i=2j,j\geq4}\left(\frac{1}{8(i-3)!}-\frac{1}{4(i-2)!}+\frac{1}{4(i-1)!(i-1)}\right)\bigg].
\end{array}$$

Therefore, Theorem \ref{theorem: main result} holds for all $n\geq
7$.

{\bf \noindent Acknowledgements}

The authors wish to express their appreciation to the referee for
his or her valuable comments and suggestions, which have led to
improvements in the presentation of this paper.

\end{document}